\documentclass[12pt,british,refpage,intoc,bibliography=totoc,index=totoc,BCOR=7.5mm,captions=tableheading]{extarticle}
\usepackage{mathptmx}
\usepackage[T1]{fontenc}
\usepackage[latin9]{inputenc}
\usepackage[a4paper]{geometry}
\usepackage{color}
\usepackage{babel}
\usepackage{mathrsfs}
\usepackage{amsmath}
\usepackage{amsthm}
\usepackage{amssymb}

\usepackage[numbers]{natbib}
\usepackage[unicode=true,pdfusetitle,
 bookmarks=true,bookmarksnumbered=true,bookmarksopen=false,
 breaklinks=false,pdfborder={0 0 0},pdfborderstyle={},backref=false,colorlinks=true]
 {hyperref}
\hypersetup{
 linkcolor=black, citecolor=blue, urlcolor=black, filecolor=blue, pdfpagelayout=OneColumn, pdfnewwindow=true, pdfstartview=XYZ, plainpages=false}

\makeatletter
\numberwithin{equation}{section}
\numberwithin{figure}{section}
\theoremstyle{plain}
\newtheorem{thm}{\protect\theoremname}[section]
\theoremstyle{definition}
\newtheorem{defn}[thm]{\protect\definitionname}
\theoremstyle{plain}
\newtheorem{lem}[thm]{\protect\lemmaname}
\theoremstyle{remark}
\newtheorem{rem}[thm]{\protect\remarkname}

\@ifundefined{date}{}{\date{}}
\usepackage{caption}
\usepackage[nottoc]{tocbibind}
\allowdisplaybreaks[4]
\usepackage{dsfont}
\DeclareMathAlphabet{\mathcal}{OMS}{cmsy}{m}{n}

\makeatother

\providecommand{\definitionname}{Definition}
\providecommand{\lemmaname}{Lemma}
\providecommand{\remarkname}{Remark}
\providecommand{\theoremname}{Theorem}

\begin{document}
\global\long\def\Sgm{\boldsymbol{\Sigma}}%

\global\long\def\W{\boldsymbol{W}}%

\global\long\def\H{\boldsymbol{H}}%

\global\long\def\P{\mathbb{P}}%

\global\long\def\Q{\mathbb{Q}}%

\title{Dirac operators and field equations\\ of the gravitational field and matter fields}
\author{Zhongmin Qian\thanks{Exeter College, University of Oxford, Oxford OX1 3DP. Email: \protect\href{mailto:zhongmin.qian@exeter.ox.ac.uk}{zhongmin.qian@exeter.ox.ac.uk}}}
\maketitle
\begin{abstract}
Dirac operators on curved space-times are introduced with the help
of a new point-view that observers have to be included in the formulation
of natural laws. The class of Dirac operators are Lorentz invariant
in the sense that the transformation rule is specified under diffeomorphisms
of the space-time which preserve the time orientation and the gravitational
field. Moreover these Dirac operators, like the original Dirac's operator
with the special relativity, satisfy the Hamiltonian relation required
by the general theory of relativity up to a correction due to the setup of a reference frame. In order to generalise (or discover)
Dirac operators, which have been an important building block in the
quantum field theories (QFTs), to curved space-times, we bring observers (which
are elements of the principal orthonormal bundle over the space-time
with its structure group being the proper Lorentz group) into the
description of Fermion fields in the presence of a gravitational field. 
As a consequence field equations combining the gravitational field and matter fields may be formulated. This work suggests that 
observational effects (due to the setup of references) in the presence of a strong gravitational 
field, i.e. matter, may be unavoidable.

\medskip

\emph{Subject areas}: QFT on curved space-times, Gravitation

\end{abstract}

\section{Introduction}

The Dirac operator introduced in \citep{Dirac1928}
is a key component in the Largangians (see Section 5 below) which
define the standard model (cf. Glashow \citep{Glashow1961}, Salam
\citep{Salam1968}, Weinberg \citep{Weinberg1967}, and \citep{Weinberg-Fields}, \citep{GiuntiKim2007} for more details). Although at the time of the discovery
of his relativistic wave equation, Dirac (cf. \citep{Dirac-principles})
did not feel that there is a need to include the gravitational
field in the wave equation for electron, there has been continuing
efforts to generalise Dirac's equation to curved space-times, for example  Schr\"odinger \citep{Schrodinger1929,Schrodinger1932},
Fock \citep{Fock1929}, Bargmann \citep{Bargmann1932}, Infeld and
van der Waerden \citep{InfeldWaerden1933}, Gursey-Lee \citep{GurseyLee1963},
and Dirac himself \citep{Dirac1935,Dirac1936}. Cartan \citep{Cartan1937}
made a systematic analysis of Dirac's equation from a point-view of
spinor representation of the Lorentz group, and concluded that in
general it is impossible to define a Dirac operator which is invariant
under the isometry group (generalised Lorentz invariance) of  curved
space-times.

The Klein-Gordon equation describing fields with integer spin, discovered by Schr\"odinger \citep{Schrodinger1926},
Gordon \citep{Gordon1926}, Fock \citep{Fock1926}, Klein \citep{Klein1926}
and de Donder and van Dungen \citep{deDonder1926} independently,
may be formulated on curved
space-times. The Dirac operator
on curved space-times was introduced in \citep{Fock1929,Schrodinger1932}
(cf. for details \citep{BirrelDavis1982,Fulling1989,ParkerToms2009})
by using the notion of spinor connections, which can be constructed
by using the tetrad formalism (cf. Weyl \citep{Weyl1950} and
the recent review \citep{Pollock2010}). Dirac operators introduced in literature
are therefore defined only locally, and a global construction relies however on the existence of  spin structures on the space-time in question. 

In complying with the special relativity, Dirac's operator is defined
as $\gamma^{\mu}D_{\mu}$, where $\gamma=(\gamma^{\mu})$ is a representation
of the Gamma matrices and $D_{\mu}$ is a connection gauging interaction
vector fields. It is crucial that the Dirac operator is invariant
under the proper Lorentz group $\mathscr{L}_{0}$ and obeys the transformation
rule for the change of the representation $\gamma=(\gamma^{\mu})$
of gamma matrices (cf. \citep[Ch. XX, pages 896-904]{MessiahII}).
It is worthy pointing out that the breakdown
of invariance of the nature laws, such as the violation of parity\footnote{The parity transformation in weak interaction is implemented by transforming
the Gamma matrices in the Dirac operator under the space inversion
operation.} \citep{LeeYang1956,Wu1957}, and the CP violation (cf. \citep{Landau1957,Christenson1964}
and \citep{Bigi-Sanda2009,Branco1999} for detail)\footnote{The CP transformation in QFT is implemented by changing the Gamma
matrices and sign of the charges coupled with vector fields that describe
interactions in the Dirac operator too.} have their origin in the breakdown of the invariance of the
Lagrangians involving Dirac operators under the full group of isometrics
of the space-time. Let us stress that the transformation rules of the Dirac
operator are vital for constructing successful field theories. 

Recall that QFT invoke mainly two
kinds of fields, the tensor type fields including gauge fields (Boson
fields) and Dirac fields (Fermion fields). The Dirac operators come
to play an important role\footnote{In QFT, Lagrangians are constructed by coupling tensor type fields
with gauges and the Fermion fields with vector fields via the Dirac
operators. The part of tensor type fields can be made Lorentz invariant on curved space-time.} for coupling gauge fields with the Fermion fields. Successful
theories have to satisfy the Lorentz invariance, which seems an obstacle
for Fermion fields in the presence of gravitational field. In the
general theory of relativity, Lorentz invariance means that theories
have to be invariant under any diffemorphisms of the space-time which
preserve both the time orientation and the gravitational field. The
existing literature on QFT on curved space-times (cf. \citep{BirrelDavis1982,ParkerToms2009,Wald 1994,Wald2006})
provides us with a scheme based on the (local) tetrad formalism.
The 
Lorentz invariance in this context has not been explored to the best knowledge of
the present author. On the other hand, successful theories are developed
over higher dimensional  space-times (of more than 5 dimensions), where the mystery
of the additional six dimensions in string theories for example, and
additional particles are not explained satisfactorily. There are still
other very interesting approaches in this direction proposed recently, for example, 
Oppenheim \citep{Oppenheim2023,OppenheimAND2023}.

To define global invariant Dirac operators on curved space-times,
there is no need to depart too far away from the accepted theoretical
frameworks, though a new, and perhaps controversial, point-view has
to be introduced. In the past \emph{reference frames}, or \emph{observers}, are not allowed to
play a role in natural laws\footnote{The common statement on the relativity principle says that the natural
laws should take the \emph{same form} mathematically, independent of reference frames. We argue
that this statement is subject to be debated and what is the same
form under different references is not well defined.}. It was commonly assumed that the idea of relativity seems particularly
design to address the formulation of natural laws free of references.
However, like in quantum physics, measurements can alter observations,
we may argue that observers can play a role in the \emph{description}
of the natural laws, and therefore observers should not be expelled
from the formulation of the natural laws. The relativity principle
demands for not the same mathematical forms of natural laws but for
an agreement of the descriptions of natural laws with respect to different
observers. 

By putting forward our point-view (we hope it is not so controversial),
we would like to outline the technical aspect how to implement our
program which is based on the global version
the tetrad formalism, that is, the frame bundle over a space-time. The tetrad formalism has been developed
into a powerful tool (cf.
\citep{Penrose1960,NewmanPenrose1962,Penrose1966,NewmanPenrose1968,PenroseMacCallum1972},
and \citep{HawkingEllis1973,PenroseRindler1984,Chandrasekhar1992}). While the tetrad
formalism has been used as a powerful computational tool, but has
not been used, to the best knowledge of the present author, as a tool
to advance new scientific concepts. 
Technically we introduce the concept of the
space-time structure which consists of the space-time continuum equipped with
a gravitational field and a time-oriented chat system. We construct
the orthonormal bundle $O^{\uparrow}(\mathscr{M})$ over the space-time
structure whose structure group being the proper Lorentz group. $O^{\uparrow}(\mathscr{M})$
is the space of events and reference frames, which has ten dimensions, perhaps by chance. We
then reveal a new invariance property of the Dirac equation introduced
originally by Dirac, which opens the route towards a construction
of Dirac operators on curved space-times. The class of Dirac operators are invariant under
transformations preserving the space-time structure and satisfy the
general relativity principle. Finally the general field equations which
involve both the gravitational field and various matter fields
may be formulated.

\section{The notion of oriented space-times}

By including observers in space-time structures, a concept of oriented space-times may be introduced naturally. 
Let
$\mathscr{M}$ be a manifold of four dimensions equipped
with a gravitational field $g$, a symmetric tensor field with signature $(1,-1,-1,-1)$. 
$g$ has to satisfy the Einstein equation
(cf. Einstein \citep{Einstein1916}, Hawking and Ellis \citep{HawkingEllis1973}),
which is however not relevant in this work. $(g_{\mu \nu})$ gives rise to a unique torsion-free metric connection, the Levi-Civita connection $\nabla$ 
defined by the Lorentz metric $g$. To take into account of the causal structure, the following constraint
is imposed on the differentiable structure of $\mathscr{M}$. There
is a time-orientated local coordinate system $(U_{i},\varphi_{i})$
on $\mathscr{M}$, where $\{U^{i}\}$ is an open cover of $\mathscr{M}$, such that for each $i$ the metric $(g_{\mu\nu})$ with respect to the
coordinate system $(x^{\mu})$ induced by $\varphi_{i}$ has the same
signature $(1,-1,-1,-1)$, in the sense that for every $p\in\mathscr{M}$,
there is a chart $(U_{i},\varphi_{i})$, so that the matrix $(g_{\mu\nu})$ at this point $p$ in this chart is $\textrm{diag}(1,-1,-1,-1)$.
If two charts $(U_{i},\varphi_{i})$ and $(U_{j},\varphi_{j})$
are overlap, on $V_{ij}=\varphi_{i}(U_{i})\bigcap\varphi_{j}(U_{j})$
the gravitational field $g$ is written as $g^{i}=(g_{\mu\nu}^{i})$
and $g^{j}=(g_{\mu\nu}^{j})$ in the chart $(U_{i},\varphi_{i})$
and $(U_{j},\varphi_{j})$ respectively, then $g^{j}=P^{-1}g^{i}P$
for some positive symmetric matrix valued function $P$
on $V_{ij}$. 

The triple of $\mathscr{M}$, $g$ and
the time-orientated charts $(U_{i},\varphi_{i})$ is called an \emph{oriented space-time}. In what follows, we shall work with a fixed oriented space-time
and shall use a specified time-orientated local
charts in our computations below.

To describe events and references over an oriented space-time, 
it is convenient to use a formulation of the frame bundle (cf. \citep{Kerner1981}and \citep{KobayashiNomizu-vol1}),
which is a global version of the tetrad formalism (cf. \citep{NewmanPenrose1962,NewmanPenrose1968}
). The frame bundle $L(\mathscr{M})$ of linear basis of $T\mathscr{M}$ over $\mathscr{M}$
is a principal fibre bundle with its structure group $\textrm{GL}(\mathbb{R}^{4})$,
and $\pi:L(\mathscr{M})\rightarrow\mathscr{M}$ the natural projection.
$L(\mathscr{M})$ is the collection of $(p;e_{0},e_{1},e_{2},e_{3})$
(denoted also by $(p;e_{\mu})$ or by $(p;e)$ if no confusion arises),
where $(e_{\mu})$ is a basis of $T_{p}\mathscr{M}$ and $p\in\mathscr{M}$. The system of local coordinate system $(x^{\mu})$
on $\mathscr{M}$ gives rise to local coordinates $(x^{\mu};x_{\;\nu}^{\sigma})$
on $L(\mathscr{M})$ determined as the following.
If $u=(p;e_{\mu})\in L(\mathscr{M})$, and $(x^{\mu})$ is a local
chart about $p\in\mathscr{M}$, then\footnote{In the tetrad formalism the $x_{\;\nu}^{\sigma}$ are chosen as
functions of $(x^{\mu})$, which are defined locally, so that $(e_{\mu})$
is a local orthonormal frame field.}
\[
e_{\mu}=\sum_{\nu}x_{\;\mu}^{\nu}\frac{\partial}{\partial x^{\nu}}\quad\textrm{ for }\mu=0,1,2,3.
\]
The matrix $(x_{\;\mu}^{\nu})\in\textrm{GL}(\mathbb{R}^{4})$, whose
inverse shall be denoted by $(x_{\nu}^{\;\mu})$. $\textrm{GL}(\mathbb{R}^{4})$
acts on $L(\mathscr{M})$ naturally:
\[
a:(p;e_{\mu})\mapsto(p;\sum_{\nu}a_{\;\mu}^{\nu}e_{\nu})
\]
for every $a=(a_{\;\mu}^{\nu})\in\textrm{GL}(\mathbb{R}^{4})$, which
preserves the fibres of $L(\mathscr{M})$: $\pi\circ a=\pi$. The
differential forms
\[
\theta^{\mu}=\sum_{\nu}x_{\nu}^{\;\mu}\textrm{d}x^{\nu}\quad\textrm{ for }\mu=0,1,2,3,
\]
are defined globally. The vertical space of $TL(\mathscr{M})$ is
defined to be the kernel of $\theta=(\theta^{\mu})$.

Under a local chart $(x^{\mu})$ on $\mathscr{M}$ and its induced
coordinate chart $(x^{\mu};x_{\;\sigma}^{\nu})$, the connection form
of the Levi-Civita connection $\nabla$ is given by 
\[
\omega_{\beta}^{\rho}=\varGamma_{\beta\mu}^{\rho}\textrm{d}x^{\mu}\quad\textrm{ and }\quad\varGamma_{\beta\mu}^{\rho}=\frac{1}{2}g^{\rho\alpha}\left(\frac{\partial g_{\beta\alpha}}{\partial x^{\mu}}+\frac{\partial g_{\mu\alpha}}{\partial x^{\beta}}-\frac{\partial g_{\beta\mu}}{\partial x^{\alpha}}\right),
\]
and define
\[
\theta_{\sigma}^{\nu}=x_{\rho}^{\;\nu}\left(\textrm{d}x_{\;\sigma}^{\rho}+x_{\;\sigma}^{\beta}\omega_{\beta}^{\rho}\right).
\]
Then $\theta^{\mu}$ and $\theta_{\sigma}^{\nu}$ together consist
of a global frame field of $T^{\star}L(\mathscr{M})$ and therefore
determine a global frame field of $TL(\mathscr{M})$. While in this work we
only need the horizontal fundamental vector fields of $L(\mathscr{M})$,
given by 
\[
L_{\mu}=x_{\;\mu}^{k}\frac{\partial}{\partial x^{k}}-\varGamma_{\alpha\beta}^{\rho}x_{\;\mu}^{\alpha}x_{\;\sigma}^{\beta}\frac{\partial}{\partial x_{\;\sigma}^{\rho}}\quad\textrm{ for }\mu=0,1,2,3.
\]
\begin{defn}
A function $f$ on $L(\mathscr{M})$ is vertical if $L_{\mu}f=0$
identically for $\mu=0,1,2,3$.
\end{defn}

There is a convenient way to introduce the fundamental fields $L_{\mu}$.
Let $X\in T_{p}\mathscr{M}$ be a tangent vector and $u=(p;e_{\mu})\in L(\mathscr{M})$.
Choose a curve $\phi:(-\varepsilon,\varepsilon)\mapsto\mathscr{M}$
(for some $\varepsilon>0$) such that $\phi(0)=p$ and $\dot{\phi}(0)=X$.
The horizontal lifting $\Phi$ of $\phi$ is the curve in $L(\mathscr{M})$:
$\Phi(t)=(\phi(t);e_{\mu}(t))$ (for $t\in(-\varepsilon,\varepsilon)$),
where $t\mapsto e_{\mu}(t)$
is the parallel translation of $e_{\mu}$ (for $\mu=0,1,2,3)$ along
the curve $\phi$ with respect to the Levi-Civita connection $\nabla$.
That is $t\mapsto e_{\mu}(t)$ are solutions to the ordinary differential
equations $\nabla_{\dot{\phi}(t)}e_{\mu}(t)=0$ for $t\in(-\varepsilon,\varepsilon)$
and $e_{\mu}(0)=e_{\mu}$ (where $\mu=0,1,2,3$). Then $\tilde{X}=\dot{\Phi}(0)\in T_{u}L(\mathscr{M})$
is called the horizontal lifting of $X\in T_{p}\mathscr{M}$ at $u\in L(\mathscr{M})$
with $\pi(u)=p$. In particular, for each $u=(p;e_{\mu})\in L(\mathscr{M})$,
and for each $\mu=0,1,2,3$, $e_{\mu}\in T_{\pi(u)}\mathscr{M}$ has
its horizontal lifting $\tilde{e}_{\mu}\in T_{p}\mathscr{M}$ at $u$. Then
$L_{\mu}(u)=\tilde{e}_{\mu}$
(for $\mu=0,1,2,3$) at $u=(p;e_{\mu})\in L(\mathscr{M})$.
\begin{lem}\label{lemma2.2}
Under a local coordinate system $(x^{\mu};x_{\;\sigma}^{\nu})$, it
holds that
\[
\left[L_{\mu},L_{\nu}\right]=x_{\;\mu}^{\alpha}x_{\;\nu}^{\beta}R_{\alpha\beta\tau}^{\rho}x_{\;\sigma}^{\tau}\frac{\partial}{\partial x_{\;\sigma}^{\rho}}
\]
where 
\[
R_{\mu\nu\rho}^{\sigma}=\frac{\partial}{\partial x^{\nu}}\varGamma_{\mu\rho}^{\sigma}-\frac{\partial}{\partial x^{\mu}}\varGamma_{\nu\rho}^{\sigma}+\varGamma_{\mu\rho}^{\alpha}\varGamma_{\alpha\nu}^{\sigma}-\varGamma_{\nu\rho}^{\alpha}\varGamma_{\alpha\mu}^{\sigma}
\]
is the curvature tensor. In particular $\left[L_{\mu},L_{\nu}\right]$
are vertical.
\end{lem}

The most important component of the spacetime structure is the sub-bundle
of the frame bundle $L(\mathscr{M})$ with the proper Lorentz group
as its structure group. Let $(\eta_{\mu\nu})$
be the Minkowski metric on the four dimensional Euclidean space $\mathbb{R}^{4}$.
The connected component of the Lorentz group $\textrm{SO}(1,3)$ at
its unity is a sub-group, denoted by $\mathscr{L}_{0}$ (following
A. Messiah \citep[page 882]{MessiahII}), consisting of all Lorentz
matrix $\varLambda=(\varLambda_{\;\mu}^{\nu})$ where $\det\varLambda=1$
and $\varLambda_{\;0}^{0}\geq1$. Let $O^{\uparrow}(\mathscr{M})$
denote the sub-bundle of the frame bundle $L(\mathscr{M})$ consisting
of all frames $(p;e_{\mu})\in L(\mathscr{M})$ where $p\in\mathscr{M}$
and $e=(e_{0},e_{1},e_{2},e_{3})$ is a linear basis of $T_{p}\mathscr{M}$
such that $\left\langle e_{\mu},e_{\nu}\right\rangle =\eta_{\mu\nu}$.
The proper Lorentz group $\mathscr{L}_{0}$ has a natural action on
$O^{\uparrow}(\mathscr{M})$. Namely, if $\varLambda=(\varLambda_{\;\mu}^{\nu})$
is an element in $\mathscr{L}_{0}$, then the action of $\varLambda$
on $O^{\uparrow}(\mathscr{M})$ gives rise to a fibre transformation:
\[
(p;e)\mapsto (p;\varLambda e)\quad\textrm{ with }(\varLambda e)_{\mu}=\varLambda_{\;\mu}^{\nu}e_{\nu}
\]
for $\mu=0,1,2,3$, so that $\pi:O^{\uparrow}(\mathscr{M})\mapsto \mathscr{M}$
is a principal fibre bundle with its structure group $\mathscr{L}_{0}$.
It is interesting to note that $O^{\uparrow}(\mathscr{M})$ is a
differentiable manifold of \emph{ten dimensions}. 

If $X\in T_{p}\mathscr{M}$ and $u=(p;e_{\mu})\in O^{\uparrow}(\mathscr{M})$,
then its horizontal lifting $\tilde{X}\in T_{u}O^{\uparrow}(\mathscr{M})$.
Indeed if $\phi:(-\varepsilon,\varepsilon)\mapsto\mathscr{M}$ with
$\phi(0)=p$, and $\Phi(t)=(\phi(t);e_{\mu}(t))$ be its horizontal
lifting at $u=(p;e_{\mu})$, then
\begin{align*}
\frac{\textrm{d}}{\textrm{d}t}g(e_{\mu}(t),e_{\nu}(t)) & =g\left(\nabla_{\dot{\phi}(t)}e_{\mu}(t),e_{\nu}(t)\right)+g\left(e_{\mu}(t),\nabla_{\dot{\phi}(t)}e_{\nu}(t)\right)\\
 & =0
\end{align*}
which implies that $\Phi(t)\in O^{\uparrow}(\mathscr{M})$ for every
$t\in(-\varepsilon,\varepsilon)$, and therefore $\tilde{X}=\dot{\Phi}(0)\in T_{u}O^{\uparrow}(\mathscr{M})$.
In particular, the fundamental horizontal fields $L_{\mu}$ are vector
fields on $O^{\uparrow}(\mathscr{M})$. By definition
\[
\tilde{X}f(u)=\left.\frac{d}{dt}\right|_{t=0}f\circ\Phi(t)
\]
for a scalar function on $O^{\uparrow}(\mathscr{M})$.

\begin{defn}
An automorphism $F:\mathscr{M}\mapsto\mathscr{M}$
is said to preserve the time-oriented space-time structure, if the following conditions
are satisfied.

\vskip0.2truecm

1) $F$ is a diffeomorphism of $\mathscr{M}$, and $F$ leaves the
metric tensor $g$ invariant, that is, $F^{\star}g=g\circ F^{-1}$,
where $F^{\star}$ denotes the differential mapping induced by $F$.

\vskip0.2truecm

2) $F$ preserves the time-orientation in the following sense. Suppose
$(p;e)\in O^{\uparrow}(\mathscr{M})$, $(F(p),F_{\star}e)$ belongs
to $O^{\uparrow}(\mathscr{M})$ too, where $F_{\star}$ is the induced
tangent mapping and $(F_{\star}e)_{\mu}=F_{\star}e_{\mu}$ for $\mu=0,1,2,3$.
Hence $F$ induces a mapping, denoted by $F_{\flat}$, from $O^{\uparrow}(\mathscr{M})$
onto $O^{\uparrow}(\mathscr{M})$ which preserves the fibres of $O^{\uparrow}(\mathscr{M})$.
\end{defn}

\begin{lem}
Suppose $F:\mathscr{M}\mapsto\mathscr{M}$ is an automorphism
preserving the time-oriented space-time structure, and $F_{\flat}$ is the induced
mapping on $O^{\uparrow}(\mathscr{M})$. Let $X\in T_{p}\mathscr{M}$,
so that $F_{\star}X\in T_{q}\mathscr{M}$ where $q=F(p)$. Let $u=(p;e_{\mu})$
be an element of $O^{\uparrow}(\mathscr{M})$ and $F_{\flat}u=(q;F_{\star}e_{\mu})$.
If $f$ is a scalar function on $O^{\uparrow}(\mathscr{M})$,
then 
\[
(\widetilde{F_{\star}X}f)(F_{\flat}u)=(\tilde{X}f\circ F_{\flat})(u).
\]
\end{lem}

The horizontal lift of $X$ at $u$ is $\tilde{X}$ (with respect
to the Levi-Civita connection $\nabla$). We want to know what is the horizontal
lift of $F_{\star}X$ at $F_{\flat}u$. Let $\phi:(-\varepsilon,\varepsilon)\mapsto \mathscr{M}$
such that $\phi(0)=p$ and $\dot{\phi}(0)=X$. Let $\Phi(t)=(\phi(t);e_{\mu}(t))$,
where $e_{\mu}(t)$ is the parallel translation of $e_{\mu}$ at $u$
along the curve $\phi$ so that $\nabla_{\dot{\phi}(t)}e_{\mu}(t)=0$
for $t\in(-\varepsilon,\varepsilon)$. Let $\theta(t)=F(\phi(t))$
and $\Theta(t)=F_{\flat}\Phi(t)=(\theta(t);F_{\star}e_{\mu}(t))$.
Then 
\[
\nabla_{\dot{\theta}(t)}F_{\star}e_{\mu}(t)=0
\]
as $F$ preserves the gravitational field $(g_{\mu\nu})$. The previous
equality implies that $\Theta(t)$ is the horizontal lift of $\theta$,
which implies that $\widetilde{F_{\star}X}=\dot{\Theta}(0)\in T_{F_{\flat}u}O^{\uparrow}(\mathscr{M})$.
By definition
\begin{align*}
(\widetilde{F_{\star}X}f)(F_{\flat}u) & =\left.\frac{\textrm{d}}{\textrm{d}t}\right|_{t=0}f\left(\Theta(t)\right)\\
 & =\left.\frac{\textrm{d}}{\textrm{d}t}\right|_{t=0}f\left(F_{\flat}(\Phi(t)\right)\\
 & =(\tilde{X}f\circ F_{\flat})(u).
\end{align*}

\vskip0.3truecm

As a consequence, we have the following fact.
\begin{lem}
Suppose $F:\mathscr{M}\mapsto\mathscr{M}$ is an automorphism
preserving the time-oriented space-time structure. If $f:O^{\uparrow}(\mathscr{M})\mapsto\mathbb{R}$
satisfying that $\tilde{X}f=0$ for any tangent vector $X\in T\mathscr{M}$.
Then $\tilde{X}f\circ F_{\flat}=0$ for any horizontal vector field
$\tilde{X}$.
\end{lem}

\section{Dirac operators on curved space-times}

 $\gamma=(\gamma^{\mu})$ is a representation of
gamma matrices, if $\gamma^{\mu}$ are elements
in $\textrm{GL}(\mathbb{R}^{4})$ such that
\begin{equation}
\gamma^{\mu}\gamma^{\nu}+\gamma^{\nu}\gamma^{\mu}=2\eta^{\mu\nu}I\label{XX.66}
\end{equation}
for $\mu,\nu=0,1,2,3$, where $I$ is the identity matrix. According
to Pauli \citep{Pauli1936}, if $(\gamma^{\mu})$ and $(\gamma'^{\mu})$
are two representations of gamma matrices, then there is
an invertible matrix $S$, unique up to a non-zero constant, such that
$\gamma^{\mu}=S\gamma'^{\mu}S^{-1}$ for $\mu=0,1,2,3$. On the other
hand if $(\gamma^{\mu})$ is a representation of
gamma matrices and $\varLambda=(\varLambda_{\;\nu}^{\mu})$ is
a Lorentz matrix, then $\hat{\gamma}^{\mu}=\varLambda_{\;\nu}^{\mu}\gamma^{\nu}$
is also a representation of gamma matrices, hence there
is an invertible matrix $S(\varLambda)$ unique up to a non-negative
constant, such that $\varLambda_{\;\nu}^{\mu}\gamma^{\nu}=S(\varLambda)^{-1}\gamma^{\mu}S(\varLambda)$
for $\mu=0,1,2,3$. 

Let $\mathscr{G}$ denote the collection of all representations of gamma matrices,
so that $\mathscr{G}\subset\textrm{GL}(\mathbb{R}^{4})^{4}$. 

Let $M$ together with $\eta=(\eta^{\mu\nu})$ be the spacetime in
special relativity. The Dirac equation for 1/2 spin particles with
mass $m$ may be written in terms of a specified gamma matrices
representation $\tilde{\gamma}=(\tilde{\gamma}^{\mu})$ as the first
order differential equation
\[
\left[\textrm{i}\tilde{\gamma}^{\mu}D_{\mu}-m\right]\varPsi=0
\]
where $D_{\mu}=\partial_{\mu}+\textrm{i}\tilde{A}_{\mu}$, $\tilde{A}=(\tilde{A}_{\mu})$
is a connection form, i.e. a differential form of first order (or
equivalently a vector field) on $M$. Here $\varPsi$ is a 4 component
complex function, i.e. a section of $\mathbb{C}^{4}\otimes M$. 
It should be stressed that the gamma matrices $\tilde{\gamma}=(\tilde{\gamma}^{\mu})$
are constant matrices. 

In the special theory of relativity, the group of isometrics of the oriented
space-time structure can be identified with the proper Poincar\'e group,
the group generated by the proper Lorentz group $\mathscr{L}_{0}$
and translations. We shall reveal another invariance property of the
Dirac equation, which in fact motivate the approach in the paper.

Let $u=(x;e)\in O^{\uparrow}(M)$, so that
$e=(e_{\mu})$ with $e_{\nu}=e_{\;\nu}^{\mu}\partial_{\mu}$ is a reference frame at $x$. 
$(e_{\;\nu}^{\mu})$ is an element of the proper Lorentz group $\mathscr{L}_{0}$.
The fundamental horizontal vector fields $L_{\mu}(p;e)=e_{\;\mu}^{\nu}\partial_{\nu}$,
or equivalently $\partial_{\mu}=e_{\mu}^{\;\nu}L_{\nu}$, so that
the Dirac operator can be rewritten as
\begin{equation}
\mathscr{D}=\gamma^{\mu}(L_{\mu}+\textrm{i}A_{\mu})\label{Dirac-Lifting}
\end{equation}
where
\[
\gamma^{\mu}(u)=e_{\rho}^{\;\mu}\tilde{\gamma}^{\rho}\quad\textrm{ and }\quad A_{\mu}(u)=e_{\;\mu}^{\nu}\tilde{A}_{\nu}\quad\textrm{ for }u=(x;e)\in O^{\uparrow}(M).
\]
Since $(e_{\rho}^{\;\mu})$ is a proper Lorentz matrix, so that $u\mapsto\gamma(u)=\left(\gamma^{\mu}(u)\right)$
is a section of $\mathscr{G}\otimes O^{\uparrow}(M)$, and $u\mapsto A(u)$
is a section of $\textrm{Hom}(O^{\uparrow}(M),TM)$. Our key observation
is that, although $u\mapsto\gamma(u)$ depends on the reference frame
$u$ and is no longer constant, but it is vertical in the sense that
$L_{\mu}\gamma=0$ identically for all $\mu$. Therefore we shall take (\ref{Dirac-Lifting})
as the definition of the Dirac operator with a given vertical section
$\gamma$ of $\mathscr{G}\otimes O^{\uparrow}(M)$ 
and a section $A$ of $\textrm{Hom}(O^{\uparrow}(M),TM)$, 
where $A_{\mu}$ is determined by writing $A(u)=\sum_{\mu}A^{\mu}(u)e_{\mu}$
for every $u=(p;e)\in O^{\uparrow}(M)$, operating on 4 component
scalar fields on $O^{\uparrow}(M)$. 

The Lorentz invariance of Dirac operators may be generalised. Indeed  one can verify for the case of special relativity that the class of the Dirac operators $\mathscr{D}_{[\gamma,A]}$
is invariant under Poinar\'e transformations. What we are going to demonstrate is that this invariance is still valid in the context of general relativity.

We are now in a position to define Dirac operators over a spacetime
with a given oriented space-time structure $\mathscr{M}$. 
For simplicity, we define the Dirac operator
with the gauge group $U(1)$, i.e. classical fields for particles
moving in a field, 
in terms of two given fields. One is a vector field which describes
the interaction in question. In our setting, it will be described
by a section $A$ of the bundle $O^{\uparrow}(\mathscr{M})\otimes T^{\star}\mathscr{M}$,
 identified with a mapping from $O^{\uparrow}(\mathscr{M})$
to $T\mathscr{M}$ such that $A(u)\in T_{\pi(u)}\mathscr{M}$ for every
$u\in O^{\uparrow}(\mathscr{M})$. At each $u=(p;e)\in O^{\uparrow}(\mathscr{M})$,
we may write $A(u)=\sum_{\mu}A^{\mu}(u)e_{\mu}$. The components
$A^{\mu}(u)$ are therefore defined. Let $A_{\mu}(u)=\eta_{\mu v}A^{\nu}(u)$. 
\begin{rem}
Here we are departing from the convention employed in the classical
quantum mechanics. In Dirac original approach, $A$ is a vector field
(or one form) and $A_{\mu}$ are components with respect to the coordinate
partials by fixing a coordinate chart.
\end{rem}

The second ingredient is a section $\gamma=(\gamma^{\mu})$ of the
product bundle $\mathscr{G}\otimes O^{\uparrow}(\mathscr{M})$, i.e.
a $\mathscr{G}$-valued section on $O^{\uparrow}(\mathscr{M})$, which
is vertical in the sense that $\tilde{X}\gamma^{\mu}=0$ identically
for $\mu=0,1,2,3,$ and for any horizontal vector field $\tilde{X}$
on $O^{\uparrow}(\mathscr{M})$.
\begin{defn}
Given a section $A$ of $O^{\uparrow}(\mathscr{M})\otimes T^{\star}\mathscr{M}$
and a vertical section $\gamma=(\gamma^{\mu})$ of $\mathscr{G}\otimes O^{\uparrow}(\mathscr{M})$,
the Dirac operator $\mathscr{D}_{\left[\gamma,A\right]}$ operating
on 4 component scalar fields
is defined as the following. Suppose $\varPsi=(\varPsi_{i})_{i=1,2,3,4}$
(as a column vector) is a 4-component (scalar) field on $O^{\uparrow}(\mathscr{M})$,
so that $\varPsi_{i}$ are complex valued functions on $O^{\uparrow}(\mathscr{M})$,
then
\begin{equation}
\mathscr{D}_{[\gamma,A]}\varPsi=\gamma^{\mu}\left(L_{\mu}+\textrm{i}A_{\mu}\right)\varPsi,\label{D-1}
\end{equation}
 where $L_{\mu}$ are the fundamental horizontal vector fields on
$O^{\uparrow}(\mathscr{M})$.
\end{defn}

The fact in the following lemma shows that our definition of Dirac operators satisfies the momentum condition, to an effect of reference correction, required by the general relativity. 
\begin{lem}
Given a section $A$ of $O^{\uparrow}(\mathscr{M})\otimes T^{\star}\mathscr{M}$
and a vertical section $\gamma=(\gamma^{\mu})$ of $\mathscr{G}\otimes O^{\uparrow}(\mathscr{M})$, 
the square of the Dirac operator
\begin{equation}
\mathscr{D}_{[\gamma,A]}^{2}=\square_{A}+\frac{1}{4}\left[\gamma^{\mu},\gamma^{\nu}\right]\left[L_{\mu},L_{\nu}\right]+\frac{\textrm{i}}{4}\left[\gamma^{\mu},\gamma^{\nu}\right]\left(L_{\mu}A_{\nu}-L_{\nu}A_{\mu}\right)\label{D-square}
\end{equation}
 operating on 4 component fields, where
\begin{equation}
\square_{A}=\eta^{\mu\nu}\left(L_{\mu}+\textrm{i}A_{\mu}\right)\left(L_{\nu}+\textrm{i}A_{\nu}\right).\label{DA-D1}
\end{equation}
\end{lem}

The equality (\ref{D-square}) follows from the well-known identity
for gamma matrices: 
\[
\gamma^{\mu}\gamma^{\nu}=\frac{1}{2}\left[\gamma^{\mu},\gamma^{\nu}\right]+\eta^{\mu\nu}I
\]
and the assumption that $\gamma^{\mu}$ are vertical so that $\gamma^{\mu}$
and $L_{\nu}$ commute: $\left[\gamma^{\mu},L_{\nu}\right]=0$. Hence
\begin{align*}
\mathscr{D}_{[\gamma,A]}^{2}\varPsi & =\gamma^{\mu}\gamma^{\nu}\left(L_{\mu}+\textrm{i}A_{\mu}\right)\left(L_{\nu}+\textrm{i}A_{\nu}\right)\varPsi\\
 & =\eta^{\mu\nu}\left(L_{\mu}+\textrm{i}A_{\mu}\right)\left(L_{\nu}+\textrm{i}A_{\nu}\right)\varPsi\\
 & +\frac{1}{4}\left[\gamma^{\mu},\gamma^{\nu}\right]\left[L_{\mu}+\textrm{i}A_{\mu},L_{\nu}+\textrm{i}A_{\nu}\right]\varPsi
\end{align*}
and
\[
\left[L_{\mu}+\textrm{i}A_{\mu},L_{\nu}+\textrm{i}A_{\nu}\right]=\left[L_{\mu},L_{\nu}\right]+\textrm{i}\left(L_{\mu}A_{\nu}-L_{\nu}A_{\mu}\right).
\]

\begin{rem}
The last term on the right-hand side of \eqref{D-square} is the term
due to the electron spin. The second term 
\[
\frac{1}{4}\left[\gamma^{\mu},\gamma^{\nu}\right]\left[L_{\mu},L_{\nu}\right]
\]
appears as the first order correction to the Klein-Gordon operator
$\square_{A}$, which has several interesting features. Firstly, the
commutate $\left[L_{\mu},L_{\nu}\right]$ is always vertical (see Lemma \ref{lemma2.2}), hence
for classical Dirac fields $\varPsi$ (i.e. fields depending only on the
spacetime variable), then this term vanishes. Secondly, the term $\left[L_{\mu},L_{\nu}\right]$
describes the curvature of the spacetime, in particular, for Minkowski
(flat) space-time, this term vanishes identically.
\end{rem}

Given a vertical section $\gamma$ of $\mathscr{G}\otimes O^{\uparrow}(\mathscr{M})$
and a section $A$ of $O^{\uparrow}(\mathscr{M})\otimes T^{\star}\mathscr{M}$,
the Dirac equation for 1/2 spin particles with mass $m$ is, by definition,
the following wave equation:
\begin{equation}
\left[\textrm{i}\gamma^{\mu}\left(L_{\mu}+\textrm{i}A_{\mu}\right)-m\right]\varPsi=0\label{Dirac-eq1}
\end{equation}
for 4 component complex field $\varPsi$ on $O^{\uparrow}(\mathscr{M})$.

Let us investigate the transformation properties of the Dirac operators. 
\begin{lem}
Let $\varLambda=(\varLambda_{\;\nu}^{\mu})$ be a section of $\mathscr{L}_{0}\otimes O^{\uparrow}(\mathscr{M})$,
i.e. an $\mathscr{L}_{0}$-valued function on $O^{\uparrow}(\mathscr{M})$.
Suppose that $\varLambda$ is vertical. Then 
\[
(\mathscr{D}_{[\gamma,A]}\varPsi)\circ\varLambda=\mathscr{D}_{[\hat{\gamma},\hat{A}]}\varPsi\circ\varLambda
\]
at $u\in O^{\uparrow}(\mathscr{M})$, where
\[
\hat{\gamma}^{\mu}=\varLambda_{\;\nu}^{\mu}\gamma^{\nu}\circ\varLambda\quad\textrm{ and }\hat{A}_{\mu}=\varLambda_{\mu}^{\;\sigma}A_{\sigma}\circ\varLambda
\]
for $\mu=0,1,2,3$.
\end{lem}

\begin{proof}
Let $u=(p;e)\in O^{\uparrow}(M)$. Then
\begin{align*}
(\mathscr{D}_{[\gamma,A]}\varPsi)_{k}(p,\varLambda e) & =(\gamma^{\mu})_{\;k}^{l}\varLambda_{\;\mu}^{\nu}(u)L_{\nu}\varPsi_{l}(p,\varLambda e)+\textrm{i}(\gamma^{\mu})_{\;k}^{l}A_{\mu}(p,\varLambda e)\varPsi_{l}(p,\varLambda e)\\
 & =(\varLambda_{\;\nu}^{\mu}\gamma^{\nu})_{\;k}^{l}L_{\mu}\varPsi_{l}(p,\varLambda e)+\textrm{i}(\varLambda_{\;\nu}^{\mu}\gamma^{\nu})_{\;k}^{l}\varLambda_{\mu}^{\;\sigma}A_{\sigma}(p,\varLambda e)\varPsi_{l}(p,\varLambda e).
\end{align*}
On the other hand, since $\varLambda$ is vertical, it holds that
\[
(L_{\mu}\varPsi\circ\varLambda)(p;e)=(L_{\mu}\varPsi)(p,\varLambda e)
\]
and therefore
\[
(\mathscr{D}_{[\hat{\gamma},\hat{A}]}\varPsi\circ\varLambda)_{k}(p,e)=\hat{\gamma}^{\mu}(p,e)\left[(L_{\mu}\varPsi)(p,\varLambda e)+\textrm{i}\hat{A}_{\mu}(p,e)\varPsi(p,\varLambda e)\right].
\]
The conclusion follows immediately.
\end{proof}
The most important issue is the invariance property of the Dirac operator.
\begin{lem}
Let $F:\mathscr{M}\rightarrow\mathscr{M}$ be an automorphism preserving
the oriented space-time structure. Then
\[
(\mathscr{D}_{[\gamma,A]}\varPsi)\circ F_{\flat}=\mathscr{D}_{[\hat{\gamma},\hat{A}]}\varPsi\circ F_{\flat}
\]
where 
\[
\hat{\gamma}^{\mu}=\gamma^{\mu}\circ F_{\flat}\quad\textrm{ and }\hat{A}_{\mu}=A_{\mu}\circ F_{\flat}.
\]
\end{lem}

In fact, $F^{\star}g=g\circ F^{-1}$, which  induces a tangent mapping
$F_{\star}$ preserving the fibres of $O^{\uparrow}(\mathscr{M})$.
$F$ therefore induces an isomorphism $F_{\flat}$ of $O^{\uparrow}(M)$,
sending $u=(p;e)$ to $F_{\flat}u=(F(p);F_{\star}e)$, where $F_{\star}e$
is the frame at $F(p)$ with components $F_{\star}e_{\mu}$. Let $\varPhi=\varPsi\circ F_{\flat}$
and $q=F(p)$. By definition
\begin{align*}
(\mathscr{D}_{[\hat{\gamma},\hat{A}]}\varPsi\circ F_{\flat})_{k}(u) & =(\hat{\gamma}^{\mu}(u))_{\;k}^{l}(L_{\mu}(\varPsi_{l}\circ F_{\flat}))(u)\\
 & +\textrm{i}(\hat{\gamma}^{\mu}(u))_{\;k}^{l}\hat{A}_{\mu}(u)(\varPsi_{l}\circ F_{\flat})(u).
\end{align*}
Since 
\[
(L_{\mu}(\varPsi_{l}\circ F_{\flat}))(u)=(\widetilde{F_{\star}e_{\mu}}\varPsi_{l})(F_{\flat}u),
\]
and therefore 
\begin{align*}
(\mathscr{D}_{[\gamma,A]}\varPsi\circ F_{\flat})_{k}(u) & =(\hat{\gamma}^{\mu}(u))_{\;k}^{l}(\widetilde{F_{\star}e_{\mu}}\varPsi_{l})(F_{\flat}u)\\
 & +\textrm{i}(\hat{\gamma}^{\mu}(u))_{\;k}^{l}\hat{A}_{\mu}(u)\varPsi_{l}(F_{\flat}u).
\end{align*}
On the other hand consider $\mathscr{D}_{[\gamma,A]}\varPsi$ at $F_{\flat}u$,
where $u=(p;e_{\mu})$. 
By definition
\begin{align*}
(\mathscr{D}_{[\gamma,A]}\varPsi)_{k}(F_{\flat}u) & =(\gamma^{\mu}(F_{\flat}u))_{\;k}^{l}(\widetilde{F_{\star}e_{\mu}}\varPsi_{l})(F_{\flat}u)\\
 & +\textrm{i}(\gamma^{\mu}(F_{\flat}u))_{\;k}^{l}A_{\mu}(F_{\flat}u)\varPsi_{l}(F_{\flat}u)
\end{align*}
so the claim follows from the definitions of $\hat{\gamma}$ and $\hat{A}$
immediately.

\vskip0.3truecm

Let $\mathfrak{L}_{\mu}=x_{\mu}^{\;\nu}L_{\nu}$ at $u=(x^{\mu},x_{\;\sigma}^{\rho})\in O^{\uparrow}(\mathscr{M})$.
Then
\[
\mathfrak{L}_{\mu}=\frac{\partial}{\partial x^{\mu}}-\varGamma_{\mu\beta}^{\rho}x_{\;\sigma}^{\beta}\frac{\partial}{\partial x_{\;\sigma}^{\rho}}\quad\textrm{ for }\mu=0,1,2,3,
\]
so that
\[
\mathscr{D}_{[\gamma,A]}=x_{\;\nu}^{\mu}\gamma^{\nu}\left(\mathfrak{L}_{\mu}+\textrm{i}x_{\mu}^{\;\sigma}A_{\sigma}\right).
\]
The computations with $\mathfrak{L}$ often take simpler form, while
we should notice that on a curved space-time, $\left(x_{\;\nu}^{\mu}\gamma^{\nu}\right)$
is no longer a representation of $\gamma$-matrices.

\section{Dirac operators in rotationally invariant space-times}

In this section we work out the Dirac operators in two important curved
space-times, and pay attention in particular to the correction
terms due to the setup of reference frames. In these models of space-times, the manifold $\mathscr{M}=\mathbb{R}^{4}$
equipped with coordinates $(t,r,\theta,\phi)\equiv(x^{\mu})$, where $t$ is the coordinate for time, while $(r, \theta, \phi)$ indicates the space variables. 
The Dirac operator with a given vertical representation $\gamma(u)$
of the gamma matrices, is defined by
\[
\mathscr{D}_{\left[\gamma,A\right]}=\gamma^{\mu}\left(L_{\mu}+\textrm{i}A_{\mu}\right)=x_{\;\nu}^{\mu}\gamma^{\nu}\left(\mathfrak{L}_{\mu}+\textrm{i}x_{\mu}^{\;\sigma}A_{\sigma}\right)
\]
at $u=(x^{\mu};x_{\;\sigma}^{\rho})$. 

\subsection{The Schwarzschild spacetime}

The gravitational field in this space-time is defined by
\[
g_{00}=f(r),\quad g_{11}=-h(r)
\]
and
\[
g_{22}=-r^{2},\quad g_{33}=-r^{2}\sin^{2}\theta\quad\textrm{ and }\quad g_{\mu\nu}=0\quad\textrm{ for }\mu\neq\nu.
\]
For this model the Christoffel symbols (cf. \citep{Chandrasekhar1992}
and \citep{Wald1984}):
\[
\varGamma_{\mu\nu}^{0}=\frac{1}{2}\frac{f'}{f}\left(\delta_{\mu1}\delta_{\nu0}+\delta_{\nu1}\delta_{\mu0}\right),
\]
\[
\varGamma_{\mu\nu}^{1}=\frac{f'}{2h}\delta_{\mu0}\delta_{\nu0}+\frac{h'}{2h}\delta_{\mu1}\delta_{\nu1}-\frac{r}{h}\delta_{\mu2}\delta_{\nu2}-\frac{r}{h}\sin^{2}\theta\delta_{\mu3}\delta_{\nu3},
\]
\[
\varGamma_{\mu\nu}^{2}=\frac{1}{r}\left(\delta_{\mu1}\delta_{\nu2}+\delta_{\nu1}\delta_{\mu2}\right)-\sin\theta\cos\theta\delta_{\mu3}\delta_{\nu3}
\]
and
\[
\varGamma_{\mu\nu}^{3}=\frac{1}{r}\left(\delta_{\mu1}\delta_{\nu3}+\delta_{\nu1}\delta_{\mu3}\right)+\frac{\cos\theta}{\sin\theta}\left(\delta_{\mu2}\delta_{\nu3}+\delta_{\nu2}\delta_{\mu3}\right).
\]
Hence 
\[
\mathfrak{L}_{0}=\frac{\partial}{\partial t}-\frac{1}{2}\frac{f'}{f}x_{\;\sigma}^{1}\frac{\partial}{\partial x_{\;\sigma}^{0}}-\frac{f'}{2h}x_{\;\sigma}^{0}\frac{\partial}{\partial x_{\;\sigma}^{1}},
\]
\[
\mathfrak{L}_{1}=\frac{\partial}{\partial r}-\frac{1}{2}\frac{f'}{f}x_{\;\sigma}^{0}\frac{\partial}{\partial x_{\;\sigma}^{0}}-\frac{1}{2}\frac{h'}{h}x_{\;\sigma}^{1}\frac{\partial}{\partial x_{\;\sigma}^{1}}-\frac{1}{r}\left(x_{\;\sigma}^{2}\frac{\partial}{\partial x_{\;\sigma}^{2}}+x_{\;\sigma}^{3}\frac{\partial}{\partial x_{\;\sigma}^{3}}\right),
\]
\[
\mathfrak{L}_{2}=\frac{\partial}{\partial\theta}+\frac{r}{h}x_{\;\sigma}^{2}\frac{\partial}{\partial x_{\;\sigma}^{1}}-\frac{1}{r}x_{\;\sigma}^{1}\frac{\partial}{\partial x_{\;\sigma}^{2}}-\frac{\cos\theta}{\sin\theta}x_{\;\sigma}^{3}\frac{\partial}{\partial x_{\;\sigma}^{3}}
\]
and
\[
\mathfrak{L}_{3}=\frac{\partial}{\partial\phi}+\frac{r}{h}\sin^{2}\theta x_{\;\sigma}^{3}\frac{\partial}{\partial x_{\;\sigma}^{1}}+\sin\theta\cos\theta x_{\;\sigma}^{3}\frac{\partial}{\partial x_{\;\sigma}^{2}}-\left(\frac{1}{r}x_{\;\sigma}^{1}+\frac{\cos\theta}{\sin\theta}x_{\;\sigma}^{2}\right)\frac{\partial}{\partial x_{\;\sigma}^{3}}.
\]

The curvature tensor
\[
R_{\mu\nu\rho}^{\sigma}=\psi^{\sigma\rho}\left(\delta_{\mu\rho}\delta_{\nu\sigma}-\delta_{\mu\sigma}\delta_{\nu\rho}\right)g_{\rho\rho},
\]
where
\[
\psi^{01}=\psi^{10}=\frac{1}{2h}\left(\frac{f''}{f}-\frac{1}{2}\frac{f'f'}{f^{2}}-\frac{1}{2}\frac{f'h'}{fh}\right),
\]
\[
\psi^{02}=\psi^{20}=\psi^{03}=\psi^{30}=\frac{1}{2rh}\frac{f'}{f},
\]

\[
\psi^{12}=\psi^{21}=\psi^{13}=\psi^{31}=-\frac{1}{2rh}\frac{h'}{h},
\]

\[
\psi^{23}=\psi^{32}=-\frac{1-h}{r^{2}h}\quad\textrm{ and }\quad\psi^{\mu\mu}=0.
\]
Hence the observers correction term 
\[
x_{\mu}^{\;\sigma}x_{\nu}^{\;\rho}\left[L_{\sigma},L_{\rho}\right]=\sum_{\sigma}\left(\psi^{\nu\mu}g_{\mu\mu}x_{\;\sigma}^{\mu}\frac{\partial}{\partial x_{\;\sigma}^{\nu}}-\psi^{\mu\nu}g_{\nu\nu}x_{\;\sigma}^{\nu}\frac{\partial}{\partial x_{\;\sigma}^{\mu}}\right).
\]

For the Schwarzschild metric, $f(r)=1-\frac{2M}{r}$ and $hf=1$,
so that
\[
\psi^{01}=\psi^{10}=-\psi^{23}=-\psi^{32}=-\frac{2M}{r^{3}},
\]
\[
\psi^{02}=\psi^{20}=\psi^{03}=\psi^{30}=\psi^{12}=\psi^{21}=\psi^{13}=\psi^{31}=\frac{M}{r^{3}},\quad\textrm{ and }\psi^{\mu\mu}=0.
\]

\subsection{The Robertson-Walker spacetime}

The Robertson-Walker metric in the spherical coordinates $(x^{\mu})=(t,r,\theta,\phi)$
is described by the following:
\[
g_{00}=1,\quad g_{11}=-\frac{a^{2}(t)}{1-Kr^{2}},\quad g_{22}=-a^{2}(t)r^{2}\quad\textrm{ and }\quad g_{33}=-a^{2}(t)r^{2}\sin^{2}\theta.
\]
Hence the Christoffel symbols
\[
\varGamma_{\mu\nu}^{0}=-f_{0}\sum_{j=1}^{3}g_{jj}\delta_{\mu j}\delta_{\nu j},
\]
\[
\varGamma_{\mu\nu}^{1}=\sum_{j=1,2,3}h_{j}g_{jj}\delta_{\mu j}\delta_{\nu j}+f_{0}\left(\delta_{\mu1}\delta_{\nu0}+\delta_{\mu0}\delta_{\nu1}\right),
\]

\[
\varGamma_{\mu\nu}^{2}=-f_{2}g^{22}g_{33}\delta_{\mu3}\delta_{\nu3}+\sum_{j=0,1}f_{j}\left(\delta_{\mu2}\delta_{\nu j}+\delta_{\mu j}\delta_{\nu2}\right)
\]
and
\[
\varGamma_{\mu\nu}^{3}=\sum_{j=0,1,2}f_{j}\left(\delta_{\mu3}\delta_{\nu j}+\delta_{\mu j}\delta_{\nu3}\right),
\]
where
\[
f_{0}(t)=\frac{a'(t)}{a(t)},\quad f_{1}(r)=\frac{1}{r},\quad f_{2}(\theta)=\frac{\cos\theta}{\sin\theta}
\]
and
\[
h_{1}=-\frac{Kr}{a^{2}(t)},\quad h_{2}=-g^{11}f_{1},\quad h_{3}=-g^{11}f_{1}.
\]
Hence
\[
\mathfrak{L}_{0}=\frac{\partial}{\partial t}-f_{0}\sum_{j=1}^{3}x_{\;\sigma}^{j}\frac{\partial}{\partial x_{\;\sigma}^{j}},
\]
\begin{align*}
\mathfrak{L}_{1} & =\frac{\partial}{\partial r}+g_{11}x_{\;\sigma}^{1}\left(f_{0}\frac{\partial}{\partial x_{\;\sigma}^{0}}-h_{1}\frac{\partial}{\partial x_{\;\sigma}^{1}}\right)-f_{0}x_{\;\sigma}^{0}\frac{\partial}{\partial x_{\;\sigma}^{1}}\\
 & -f_{1}\left(x_{\;\sigma}^{2}\frac{\partial}{\partial x_{\;\sigma}^{2}}+x_{\;\sigma}^{3}\frac{\partial}{\partial x_{\;\sigma}^{3}}\right),
\end{align*}
\begin{align*}
\mathfrak{L}_{2} & =\frac{\partial}{\partial\theta}+g_{22}x_{\;\sigma}^{2}\left(f_{0}\frac{\partial}{\partial x_{\;\sigma}^{0}}-h_{2}\frac{\partial}{\partial x_{\;\sigma}^{1}}\right)\\
 & -\left(f_{0}x_{\;\sigma}^{0}+f_{1}x_{\;\sigma}^{1}\right)\frac{\partial}{\partial x_{\;\sigma}^{2}}-f_{2}x_{\;\sigma}^{3}\frac{\partial}{\partial x_{\;\sigma}^{3}}
\end{align*}
and
\begin{align*}
\mathfrak{L}_{3} & =\frac{\partial}{\partial\phi}+g_{33}x_{\;\sigma}^{3}\left(f_{0}\frac{\partial}{\partial x_{\;\sigma}^{0}}-h_{3}\frac{\partial}{\partial x_{\;\sigma}^{1}}+f_{2}g^{22}\frac{\partial}{\partial x_{\;\sigma}^{2}}\right)\\
 & -\left(\sum_{j=0,1,2}f_{j}x_{\;\sigma}^{j}\right)\frac{\partial}{\partial x_{\;\sigma}^{3}}.
\end{align*}
The curvature tensor
\[
R_{\mu\nu\rho}^{\sigma}=\psi^{\sigma\rho}\left(\delta_{\mu\sigma}\delta_{\nu\rho}-\delta_{\mu\rho}\delta_{\nu\sigma}\right)
\]
where
\[
\psi^{0\sigma}=(1-\delta_{0\sigma})\frac{a''(t)}{a(t)}g_{\sigma\sigma},\quad\textrm{ for }\sigma=0,1,2,3.
\]
and
\[
\psi^{i0}=(1-\delta_{i0})\frac{a''(t)}{a(t)}g_{00},\quad\psi^{ij}=(1-\delta_{ij})\frac{a'(t)^{2}+K}{a(t)^{2}}g_{jj}\textrm{ for }i,j=1,2,3.
\]
The observers corrections are given as the following.
\[
x_{0}^{\;\sigma}x_{i}^{\;\rho}\left[L_{\sigma},L_{\rho}\right]=\frac{a''(t)}{a(t)}\left(g_{ii}x_{\;\sigma}^{i}\frac{\partial}{\partial x_{\;\sigma}^{0}}-g_{00}x_{\;\sigma}^{0}\frac{\partial}{\partial x_{\;\sigma}^{i}}\right)
\]
for $i=1,2,3$ and 
\[
x_{i}^{\;\sigma}x_{j}^{\;\rho}\left[L_{\sigma},L_{\rho}\right]=\frac{a'(t)^{2}+K}{a(t)^{2}}\left(g_{jj}x_{\;\sigma}^{j}\frac{\partial}{\partial x_{\;\sigma}^{i}}-g_{ii}x_{\;\sigma}^{i}\frac{\partial}{\partial x_{\;\sigma}^{j}}\right)
\]
for $1\leq i<j\leq3$.

\section{Field equations}

In this section we propose field equations which include the gravitational
field and matter fields as well, based on the concept of space-time
structures introduced in the previous sections. The space-time is
modelled according to Einstein by a continuum $\mathscr{M}$ a connected
manifold of four dimensions, equipped with a local coordinate system
$(x^{\mu})$ which preserves the time-orientation, in the sense that
the collection $\mathscr{G}(\mathscr{M})$ of Lorentz metrics on $\mathscr{M}$
compatible with the time-orientation is not empty. The frame bundle
over $\mathscr{M}$ is denoted by $L(\mathscr{M})$ which has the
induced coordinate system $(x^{\mu},x_{\;\rho}^{\sigma})$ or simply
written as $(x,X)$. If $g\in\mathscr{G}(\mathscr{M})$ then $\textrm{d}s^{2}=g_{\mu\nu}\textrm{d}x^{\mu}\textrm{d}x^{\nu}$
where $g_{00}>0$, which in turn defines a natural volume measure
$V_{g}(u,\textrm{d}u)$ on $O^{+}(\mathscr{M})$. This measure $V_{g}(u,\textrm{d}u)$
depends only on the Lorentz metric $g=(g_{\mu\nu})$ and can be split
as a (quasi) product measure. Indeed if $u=(x,e)$ represents a general
element of $O^{+}(\mathscr{M})$ then
\[
V_{g}(u,\textrm{d}u)=K_{g(x)}(x,\textrm{d}e)\sqrt{-\det g(x)}\textrm{d}x
\]
where of course $\sqrt{-\det g}\textrm{d}x$ is the volume measure
on $\mathscr{M}$. $K_{g(x)}(x,\textrm{d}e)$ is a kernel so that it is
a measure for every fixed $x$ on the fibre $O_{x}^{+}(\mathscr{M})$
which can be described as the following.

Let $x\in\mathscr{M}$ with coordinates $(x^{\mu})$. Then $e=(x_{\;\rho}^{\sigma})$
belongs to $O_{x}^{+}(\mathscr{M})$ if and only if 
\begin{equation}
g_{\sigma\rho}(x)x_{\;\mu}^{\sigma}x_{\;\nu}^{\rho}=\eta_{\mu\nu}\label{G-c1}
\end{equation}
(10 independent equations) so that $O_{x}^{+}(\mathscr{M})$ is identified
with a six dimensional sub-manifold of $\mathbb{R}^{4\times4}$
equipped with the standard metric and therefore with the Lebesgue
measure. Therefore $O_{x}^{+}(\mathscr{M})$ carries an induced Riemannian
metric from $\mathbb{R}^{4\times4}$ and  the induced volume
measure which is $K_{g(x)}(x,\textrm{d}e)$.

In order to write down the field equations, a few notations
have to be introduced. Suppose $A=(A_{\sigma\rho})$ is a $4\times4$
non-degenerate and symmetric matrix. Consider the sub-manifold $N$
of $\mathbb{R}^{4\times4}$ determined by the ten quadratic equations:
\begin{equation}
A_{\sigma\rho}x_{\;\mu}^{\sigma}x_{\;\nu}^{\rho}=\eta_{\mu\nu}\quad\textrm{ for }\mu\geq\nu.\label{A-01}
\end{equation}

By symmetry, we may use $(x_{\;\rho}^{\sigma})_{\sigma<\rho}$ as
coordinates for the sub-manifold $N$ of six dimensions. By solving
$x_{\;\nu}^{\mu}$ for $\mu\geq\nu$ from (\ref{A-01}) in terms of
$x_{\;\rho}^{\sigma}$ (where $\sigma<\rho$) we may write 
\[
x_{\;\nu}^{\mu}=f_{\;\nu}^{\mu}(x_{\;\rho}^{\sigma}:\sigma<\rho)=f_{\;\nu}^{\mu}(X:<)
\]
where we have used the short notation $X:<$ to denote $(x_{\;1}^{0},x_{\;2}^{0},x_{\;3}^{0},x_{\;2}^{1},x_{\;3}^{1},x_{\;3}^{2})$.
The induced Riemann metric on $N$ is given by
\[
G_{(\sigma<\rho),(\alpha<\beta)}(A)=\delta_{(\sigma<\rho),(\alpha<\beta)}+\sum_{\mu\geq\nu}\frac{\partial f_{\;\nu}^{\mu}}{\partial x_{\;\rho}^{\sigma}}\frac{\partial f_{\;\nu}^{\mu}}{\partial x_{\;\beta}^{\alpha}}
\]
and the volume measure on $N$ has an explicit representation.
\[
\sqrt{\det\left(G_{(\sigma<\rho),(\alpha<\beta)}(A)\right)}\prod_{\mu<\nu}\textrm{d}x_{\;\nu}^{\mu}.
\]
The important fact to us is that the following mapping 
\[
A\rightarrow\log\det\left(G_{(\sigma<\rho),(\alpha<\beta)}(A)\right)
\]
is differentiable and 
\[
\left.\frac{d}{d\varepsilon}\right|_{\varepsilon=0}\log\sqrt{\det G(A(\varepsilon))}=\left(\frac{1}{2}D\log\det G(A)\right)(\delta A)
\]
which depends on the variation $\delta A$ linearly. 

We are now in a position to write down the field equations involving
both the gravitational field and matter fields. 

The action functional according to quantum field theories may be written
as
\[
L(g,\varPsi,\varPhi,D_{[B]}\Phi,\mathscr{D}_{[\gamma,A]}\varPhi,A,B)
\]
which is defined in terms of the following integral
\[
\int_{O^{+}(\mathscr{M})}\left(R(g)+\mathscr{L}(g,\varPsi,\varPhi,D_{[B]}\Phi,\mathscr{D}_{[\gamma,A]}\varPhi,A,B)\right)V_{g}(u,\textrm{d}u)
\]
where $g$ varies in $\mathscr{G}(\mathscr{M})$, $\varPsi$ and $\varPhi$
represent tensor type fields and spinor fields, $B$ represents gauge
fields and $D_{[B]}$ the corresponding co-variant derivative, $\mathscr{D}_{[\gamma,A]}$
are Dirac operators and $A$ are vector fields. However we do not
consider the representation of Gamma matrix $\gamma$ as a field and
therefore $\gamma$ is fixed. In this expression, $R(g)$ is the scalar
curvature of $g$ and 
\[
\mathscr{L}(g,\varPsi,\varPhi,D_{[B]}\Phi,\mathscr{D}_{[\gamma,A]}\varPhi,A,B)
\]
is a Lagrangian density with respect to the volume measure on ${O^{+}(\mathscr{M})} $. The action functional defines the field equations:
\[
\frac{\delta}{\delta g}L(g,\varPsi,\varPhi,D_{[B]}\Phi,\mathscr{D}_{[\gamma,A]}\varPhi,A,B)=0,
\]
\[
\frac{\delta}{\delta\varPsi}L(g,\varPsi,\varPhi,D_{[B]}\Phi,\mathscr{D}_{[\gamma,A]}\varPhi,A,B)=0
\]
and etc., where the first equation, the field equation of variations
of Lorentz metric $g$ is the generalization of Einstein's gravitational field
equation. 

This generalisation of Einstein's field equation without matter fields
(i.e. for the case where the Lagrangian density $\mathscr{L}$ vanishes
identically), under the local coordinate system $(x,X)=(x^{\mu},x_{\;\rho}^{\sigma})$
for $\sigma<\rho$, is given as the field equation:
\[
R_{ab}-\frac{1}{2}Rg_{ab}+\frac{1}{2}RD\log\det G(g,X<)_{ab}=0,
\]
here the additional term appears as the observational correction term. 

Thanks to the concept of the space-time structure, as if it is given by the lord nature, the field equations proposed in this article, involve the co-variant derivatives of the space-time variables, but do not involve any frame dynamics or dynamics of the gamma matrices, or otherwise one must give up the idea of references entering our formulation of natural laws.

\end{document}